\documentclass[12pt,reqno,a4paper]{amsart}
\usepackage{graphics,graphicx}
\usepackage{amssymb}
\usepackage{enumerate}

\usepackage{amsmath}

\usepackage{color}
\usepackage{colordvi}

\numberwithin{equation}{section}

 \theoremstyle{plain}            
 \newtheorem{theorem}{Theorem}[section]
 \newtheorem{proposition}[theorem]{Proposition}
 \newtheorem{lemma}[theorem]{Lemma}

 \theoremstyle{definition}       

 \newtheorem{example}[theorem]{Example}

 \newtheorem{remark}[theorem]{Remark}

\newcommand{\ee}{\mathrm{e}}
\newcommand{\D}{\mathrm{d}}

\newcommand{\N}{\mathbb{N}}
\newcommand{\R}{\mathbb{R}}

\newcommand{\DD}{\mathcal{D}}
\newcommand{\OO}{\mathcal{O}}

\newcommand{\defeq}{\vcentcolon=} 

\usepackage{mathtools}

\begin{document}
\title[Magnetic transport due to a potential obstacle]
{Magnetic transport due to a translationally invariant potential obstacle}
\author[P.\ Exner]{Pavel Exner}
\address{Doppler Institute for Mathematical Physics and Applied Mathematics,  Czech Technical University,
B\v rehov\'a 7, 11519 Prague, Czechia, and Department of Theoretical Physics, NPI, Czech Academy of Sciences, 25068 \v{R}e\v{z}
near Prague, Czechia} \email{exner@ujf.cas.cz}
\urladdr{http://gemma.ujf.cas.cz/~exner/}

\author[D.\ Spitzkopf]{David Spitzkopf}
\address{Department of Theoretical Physics, NPI, Czech Academy of Sciences, 25068 \v{R}e\v{z} near Prague, Czechia, and
Faculty of Mathematics and Physics, Charles University, V Hole\v{s}ovi\v{c}k\'ach 2, 18000 Prague, Czechia}
\email{spitzkopf98@gmail.com}

\maketitle
\begin{abstract}
We consider a two-dimensional system in which a charged particle is exposed to a homogeneous magnetic field perpendicular to the plane and a potential that is translationally invariant in one dimension. We derive several conditions on such a perturbation under which the Landau levels change into an absolutely continuous spectrum.
\end{abstract}

\section{Introduction} 
\setcounter{equation}{0}

The topic of this paper is the motion of a two-dimensional charged particle exposed to a homogeneous magnetic field of intensity $B$ perpendicular to the plane. In the absence of perturbations, it is described by the Landau Hamiltonian,
\begin{equation} \label{LandauH}
    H_{L} \defeq \left( -i \partial_x + A_x\right)^2 + \left( -i \partial_y + A_y\right)^2\mbox{ on } L^2\left(\R^2 \right),
\end{equation}
with the domain of definition $\mathcal{D}_L = \lbrace \psi \in L^2\left(\R^2 \right):\: \left(H_L \psi \right)(x) \in  L^2\left(\R^2 \right) \rbrace$, where $H_L \psi$ is understood in the sense of distributions. The vector potential can be chosen in different ways; for our purposes it is useful to work with the Landau gauge,
\begin{equation} \label{LandauG}
    \textbf{A}=\left(0,Bx\right).
\end{equation}
Magnetic field, in particular, a homogeneous one, has localizing effect both classically and quantum mechanically; it is well known that the spectrum of the operator \eqref{LandauH} is pure point, consisting of infinitely degenerate eigenvalues, the \emph{Landau levels}, $\epsilon_n=B(2n+1),\, n=0,1,2,\dots$. Perturbations of $H_L$ remove this degeneracy, partially or fully, but as long as they are of a localized character, the point nature of the spectrum is preserved \cite{LR15, BEHL20}; to achieve a transport associated with the presence of an absolutely continuous component of the spectrum, the perturbation must be infinitely extended.

The analysis simplifies considerably if the perturbation is \emph{translationally invariant} since with help of the gauge \eqref{LandauG} one is then able to pass to a unitarily equivalent operator in the form of a direct integral the fibers of which are one-dimensional Schr\"odinger operators. A well-known example is the restriction of the Landau Hamiltonian \eqref{LandauH} to a halfplane with Dirichlet \cite{DP99}, Neumann \cite[Sec.~2.4]{Ra17}, or Robin \cite{Gr21} boundary; the halfplane boundary can be replaced by a potential wall \cite{MMP99}. Another solvable model of this type involves perturbation of the operator \eqref{LandauH} in the plane by a singular $\delta$-type potential supported by a line \cite{Gr21}.

Magnetic transport due to translationally invariant perturbations can also be considered beyond such explicit models. The best known example is represented by the \emph{Iwatsuka-type systems} \cite{Iw85} in which the perturbation is magnetic, the vector potential \eqref{LandauG} being replaced by $(0,Bx+a(x))$ with some $a:\R\to\R$. It was conjectured \cite[Sec.~6.5]{CFKS87} that any nontrivial perturbation of this type makes the spectrum purely absolutely continuous. While in general this question is still open, there are various sufficient conditions of the function $a$ under which this happens -- see, e.g., \cite{EK00, MP97, MP18}.

In this paper, we consider another type of perturbation, a translationally invariant regular potential; we ask about the spectrum of the the operator $H_L+V$ with $V(x,y)=v(x)$, where $v:\R\to\R$ is a function the properties will be specified below. This problem attracted some interest in connection with recent investigations of \emph{soft quantum waveguides}. In the absence of a magnetic field, the spectrum of Schr\"odinger operator with a straight potential channel is naturally purely absolutely continuous; it may acquire a discrete component if such a waveguide is locally bent keeping its profile \cite{Ex20, EV24, KKK21}; this discrete spectrum can be destroyed by a \emph{local} magnetic field \cite{BBS24}. In this context, one naturally asks what is the influence of a \emph{homogeneous} magnetic field on a straight potential channel spectrum.

We are going to demonstrate several conditions on the function $v$ under which the perturbation changes the Landau levels into a (fully) absolutely continuous spectrum; in analogy with the Iwatsuka model we \emph{conjecture} that this happens for any nontrivial perturbation profile. Our method is the same as in the above mentioned work, rephrasing the problem as an investigation of a family of one-dimensional Schr\"odinger operators, looking into the dependence of their eigenvalues on the momentum variable referring to the motion in the $y$ direction. Let us also recall recent related result on magnetic transport in three-dimensional systems of layer type \cite{Ex22, EKT18}.

\section{The model}
\setcounter{equation}{0}

As indicated, the first step is to pass from the operator $H_V=H_L+V$ on $L^2(\R^2)$ to a unitarily equivalent one by means of the Fourier-Plancherel operator in the $y$ variable
\begin{equation} \label{dirint}
     \hat{H}_{V} := \mathcal{F}_{y\to p}^{-1} H_V \mathcal{F}_{y \to p} = \int_{\mathbb{R}}^{\oplus} h_v(p) \D p\,;
\end{equation}
using \eqref{LandauH} and \eqref{LandauG} one can check easily that the fiber operators are
\begin{equation} \label{eq:fiberTransInvariant}
    h_v(p) \defeq -\partial_{x}^2 + (p + Bx)^2 + v(x),\;\; p\in\R,
\end{equation}
on $L^2(\R)$; we will always assume that \emph{the perturbation is nontrivial}, $v\ne 0$. The spectral properties of $H_V$ are then determined by those of the operators $h_v(p)$ \cite[Sec.~XIII.16]{RS}. Concerning the profile potential, we assume first that
 \begin{enumerate}[(v1)]
 \setlength{\itemsep}{1.5pt}
\item $v_+\in L^2_\mathrm{loc}(\R)$ and $v_-\in (L^2+L^\infty)(\R)$, where we put conventionally $v_\pm(x):=\max(\pm v(x),0)$.  \label{v1}
 \end{enumerate}
Under this condition the operator $h_v(p)$ is for any $p\in\R$ essentially self-adjoint on $C_0^\infty(\R)$ \cite[Thm.~X.29]{RS}, and moreover, we have

\begin{proposition}
\label{prop:purepoint}
The spectrum of $h_v(p)$ is purely discrete and simple consisting of eigenvalues $\epsilon_n(p),\: n=0,1,2,\dots\,$.
\end{proposition}
\begin{proof}
Was $v_-$ essentially bounded outside a compact, the discrete character of the spectrum would follow from \cite[Thm.~XIII.29]{RS}, however, we do not need this additional hypothesis. Indeed, by (v\ref{v1}) we have $v_-=v_{2,-}+v_{\infty,-}$ with $v_{2,-}\in L^2$, hence for any positive $a,\delta$ we can infer that
$$ 
\int_a^{a+\delta} \!v_{2,-}(x)\,\D x \le \delta^{1/2} \Big( \int_a^{a+\delta} \!v_{2,-}(x)^2\,\D x \Big)^{1/2} \!\le \delta^{1/2} \Big( \int_a^\infty \!v_{2,-}(x)^2\,\D x \Big)^{1/2}
$$ 
and the right-hand side tends to zero as $a\to\infty$ in view of the absolute continuity of the integral. The same is true for the integral of $v_{2,-}$ over the interval $(-a, -a-\delta)$ with any $\delta>0$ as $a\to\infty$. A simple estimate,
\begin{align*} 
\int_a^{a+\delta} & \big((p + Bx)^2 + v(x)\big)\,\D x \\ & \ge \delta(Ba+p)(Ba+p+\delta) -\frac13\delta^3 -\delta\|v_{2,\infty}\|_\infty + \int_a^{a+\delta} v_{2,-}(x)\,\D x,
\end{align*}
and its counterpart for the integral of the potential over \mbox{$(-a, -a-\delta)$} show that $\int_a^{a+\delta} v(x)\,\D x \to\infty$ holds as $|a|\to\infty$ for any $\delta>0$, hence the spectrum is purely discrete by Molchanov theorem \cite{Mo53}. The simplicity comes for the fact that \eqref{eq:fiberTransInvariant} is limit-point at both $\pm\infty$.
\end{proof}

\begin{remark} \label{rem:efs}
If needed we shall indicate the potential writing $\epsilon_n(p,v)$. The corresponding normalized eigenfunctions will be denoted $\phi_n(\cdot;p)$ or $\phi_n(\cdot;p,v)$; without loss of generality we may choose them real-valued. Unless $\epsilon_n(\cdot)$ is constant, they represent states transported along the perturbation; the corresponding probability current density is $p|\phi_n(\cdot;p)|^2$.
\end{remark}

By Theorem~XIII.86 of \cite{RS} the spectrum of $\hat{H}_{V}$, and thus also of the original operator $H_V$, will be purely absolutely continuous provided
\begin{enumerate}
    \item the family $\lbrace h_v(p):\: p \in \R \rbrace$ is analytic with respect to $p$,
    \item no eigenvalue branch $\epsilon_n(\cdot)$ is constant.
\end{enumerate}
The analyticity makes verification of the condition (ii) easy; it is enough to check that it takes two different values. The property is easy to establish, without any additional assumption on $v$.
\begin{proposition}
\label{prop:TraInvAnal}
Under assumption (v\ref{v1}) the condition (i) is satisfied.
\end{proposition}
\begin{proof}
Indeed, for any $p_0\in\R$ we have $h_v(p)=h_v(p_0)+d_p$ in the form sense, where
\begin{equation} \label{pert}
    d_p[\psi]= (p^2-p_0^2)\|\psi\|^2 + (\psi,2Bx(p-p_0)\psi).
\end{equation}
The first term on right-hand side of \eqref{pert} is bounded for $p$ in any fixed neighborhood of $p_0$, concerning the second one we note that for any $\varepsilon\in (0,1)$ one can find a $\delta>0$ such that
$$ 
    \varepsilon(p_0-Bx)^2 \ge 2Bx(p-p_0) - \delta,
$$ 
holds for all $x\in\R$, thus this term is relatively bounded with respect to $h_v(p_0)$ with the bound less than one; this means that $\{h_v(p):\: p \in \R \}$ is an analytic family of type (B) in the sense of Kato.
\end{proof}

\section{Some simple properties}
\setcounter{equation}{0}

Our main aim is to find conditions which ensure non-constancy of the functions $\epsilon_n(\cdot)$. A simple one is based on what one could call an extended asymmetry. Assume that
 \begin{enumerate}[(v1)]
 \setlength{\itemsep}{1.5pt}
 \setcounter{enumi}{1}
\item finite limits $v_{\pm} \defeq \lim_{x \to \pm \infty} v(x)$ exist.  \label{v2}
 \end{enumerate}
Needless to say, if the two limits exist and equal each other, with the spectral continuity in mind one may assume that $v_{\pm}=0$ without loss of generality. Then we have the following claim:
\begin{proposition}
\label{prop:TraInvUnevenLims}
In addition to assumption (v\ref{v1}) and (v\ref{v2}), suppose that $v_+ \neq v_- $, then no eigenvalue branch $\epsilon_n(p)$ is constant in $p$.
\end{proposition}
\begin{proof}
As we have indicated, in view of the analyticity it is enough to show that $\epsilon_n(\cdot)$ can assume different values, which follows from the fact that $\lim_{p \to \pm \infty} \epsilon_n(p) = v_{\pm}+ B(2n+1)$ holds for any $n\in\mathbb{N}_0$. Indeed, by variable shift, the spectrum of operator \eqref{eq:fiberTransInvariant} coincides with that of
\begin{equation} \label{eq:fiberTransInvariant2}
    \tilde{h}_v(p) \defeq -\partial_{x}^2 + B^2x^2 + v\big(x-\textstyle{\frac{p}{B}}\big)
\end{equation}
which can be regarded as a perturbation of the harmonic oscillator; just the $\pm$ signs of the asymptotic limits are swaped. Let us consider operators $\tilde{h}_v^{N/D}(p)$ obtained by adding Neumann/Dirichlet condition at the point $x=\frac{p}{2B}$ which allow us to use the bracketing estimate \cite[Sec.~XIII.15]{RS},
$$ 
    \tilde{h}_v^N(p) \le \tilde{h}_v(p) \le \tilde{h}_v^D(p)
$$ 
The spectrum of each of $\tilde{h}_v^{N/D}(p)$ is the union of the respective operators on $\big(-\infty,\frac{p}{2B}\big)$ and $\big(\frac{p}{2B},\infty\big)$; for the limit $p\to\infty$ only the former are relevant. As in \cite{CSH02}, one can check that they both converge to the spectrum of $\tilde{h}_v(p)$, hence to any $\varepsilon>0$ there is a $p_1>0$ such that
$$ 
    |\epsilon_n(p)- \tilde\epsilon_n^{N/D}(p)| = |\tilde\epsilon_n(p)- \tilde\epsilon_n^{N/D}(p)| < \textstyle{\frac12}\varepsilon
$$ 
holds for all $p>p_1$. Furthermore, by assumption (v\ref{v2}) there is a $p_2>0$ such that $|v_- -v(x-\frac{p}{B})|<\frac12\varepsilon$ holds on $\big(-\infty,\frac{p}{2B}\big)$ for all $p>p_2$, and consequently
$$ 
    |\tilde\epsilon_n^{N/D}(p) -B(2n+1) - v_-| < \textstyle{\frac12}\varepsilon\,;
$$ 
putting the two estimates together, we get $|\epsilon_n(p) -B(2n+1) - v_-| < \varepsilon$ for all $p>\max\{p_1,p_2\}$. The analogous estimate can be used for $v_+$ and the limit $p\to-\infty$ which establishes the claim.
\end{proof}

In the absence of such a manifested asymmetry, things are more complicated. If we strengthen assumption (v\ref{v1}), namely
 \begin{enumerate}[(v1)]
 \setcounter{enumi}{2}
 \setlength{\itemsep}{1.5pt}
\item $v\in L^2(\R)$,  \label{v3}
 \end{enumerate}
one can use first-order perturbation theory to prove the absolute continuity for weak enough potential on a finite energy interval.
\begin{proposition}
\label{prop:TraInvSmallPert}
Under assumption (v\ref{v3}), to every $n\in\N_0$ there is a $\lambda_n>0$ such that the functions $\epsilon_j(\cdot;\lambda v)$ are non-constant for all $j=0,1, \dots,n$.
\end{proposition}
\begin{proof}
As mentioned above, the eigenvalues in question coincide with those of operator \eqref{eq:fiberTransInvariant2}. First of all, we need again to establish the asymptotic behavior of the functions $\epsilon_n(\cdot)$, namely to check that in the present situation we have $\lim_{p \to \pm \infty} \epsilon_n(p) = B(2n+1)$. Since we do not assume (v\ref{v2}) now, we have to choose an argument different from that of Proposition~\ref{prop:TraInvUnevenLims}. Consider the quadratic form of operator \eqref{eq:fiberTransInvariant2},
\begin{equation} \label{eq:tildeqform}
    \tilde{q}_v(p)[\psi] = \int_\R |\psi'(x)|^2\D x + \int_\R |Bx\psi(x)|^2 \D x + \int_\R v\big(x-\textstyle{\frac{p}{B}}\big) |\psi(x)|^2 \D x.
\end{equation}
By the argument used in the proof of Proposition~\ref{prop:TraInvAnal} its domain is the same for all $p\in\R$, and there is common core, e.g., $C_0^\infty(\R)$. In view of the absolute continuity of the integral, the last term on the right-hand side tends to zero as $|p|\to\infty$ for any $\psi\in C_0^\infty(\R)$. By \cite[Thm.~VIII.3.6]{Ka} this implies $\tilde{h}_v(p)\to h_0$, the harmonic oscillator Hamiltonian, in the strong resolvent sense, in particular, the sought convergence of the eigenvalues.

Next we fix a $j$ and replace $v$ by $\lambda v$; denoting then the normalized real-valued oscillator eigenfunctions by $\varphi_j$, we get
\begin{equation} \label{first_pert}
    \epsilon_j(p;\lambda v) = B(2j+1) + \lambda \int_\R  v\big(x-\textstyle{\frac{p}{B}}\big)\, \varphi_j(x)^2\,\D x + \OO(\lambda^2)
\end{equation}
Assume that the integral in this first-order perturbation expansion term is zero for all $p\in\R$. Putting $u(z)=v(-z)$, it takes a convolution form and applying to it the Fourier-Plancherel operator we infer that the assumption is equivalent to $\big(\mathcal{F}_{x \to k}u\big)(k) (\widehat{\varphi_j^2})(k)=0$ a.e. in $\R$. The second factor on the left-hand side equals $\hat\varphi_j * \hat\varphi_j$, and from the explicit knowledge of $\varphi_j$ it is not difficult to check that as a function of $p$ it can have isolated zeros only, it is then the first factor which must vanish; using further the injectivity of $\mathcal{F}_{x \to k}$, we get $v(x)=0$ a.e. in $\R$.

This is a contradiction, however, because we assumed that the perturbation is nontrivial. This means that the integral at the right-hand side of \eqref{first_pert} is nonzero for some $p$. For weak coupling the first-order term dominates over the remainder, hence there is a $\lambda(j)>0$ such that $\epsilon_j(p;\lambda v) \ne B(2j+1)$ holds for $\lambda\in(0,\lambda(j))$. To get the claim, it then enough to put $\lambda_n:= \min_{0\le j\le n} \lambda(j)$.
\end{proof}

Before proceeding further, let us mention situations when neither (v\ref{v1}) nor (v\ref{v2}) is satisfied.

\begin{example}
Consider potentials of the form $v(x) = \alpha x - \beta^2x^2$. If $\beta=0$ and $\alpha\ne 0$, the operator describes transport in crossed electric and magnetic fields which is explicitly solvable in view of the identity
$$ 
    (Bx+p)^2+\alpha x = \big(Bx+p+\textstyle{\frac{\alpha}{2B}}\big)^2 + p^2 + \frac{\alpha p}{B} - \frac{\alpha^2}{4B^2}.
$$ 
Thus we have $\epsilon_n(p)= B(2n+1) + p^2 + \frac{\alpha p}{B} - \frac{\alpha^2}{4B^2}$ so that the spectrum is purely absolutely continuous and covers the interval $\big(B-\frac{\alpha^2}{2B^2}, \infty\big)$.

On the other hand, if $\alpha=0$ one can again use completion to the square but now the coefficient values matters:
\begin{enumerate}
    \item If $B^2-\beta^2 > 0$ the system exhibits transport along such a potential barrier: the spectrum of $h_v(p)$ is purely discrete consisting of the eigenvalues
$$ 
    \epsilon_n(p)= \sqrt{B^2-\beta^2}\,(2n+1) - \frac{B^2 p^2}{\sqrt{B^2-\beta^2}} +p^2
$$ 
which means that the spectrum of $H_V$ is purely absolutely continuous and covers the whole real line.
   \item If $B^2-\beta^2 < 0$ the last claim remains true but there is a significant difference: the fiber operator is in this case
$$ 
    -\partial_{x}^2 - \big(\sqrt{\beta^2-B^2}x - \frac{Bp}{\sqrt{\beta^2-B^2}}\big)^2 + \frac{B^2 p^2}{\sqrt{\beta^2-B^2}} +p^2.
$$ 
It is still essentially self-adjoint because solutions to the deficiency equation behave in the $|x|\to\infty$ asymptotics as $\frac{1}{\sqrt{x}}\, \ee^{\pm ix^2/2}$ \cite[Thm.~6.2.2]{Ol}, hence they are not $L^2$; the spectrum of the fiber operator is purely absolutely continuous, without eigenfunctions, in other words, there are no states which could be transported along the barrier.
\end{enumerate}
\end{example}

\section{More on the absolute continuity}
\setcounter{equation}{0}

Another way is to use the first-order perturbation theory again, this time taking $p$ as the parameter, in the spirit of Feynman and Hellmann. In this way, one is able to avoid the requirement of a potential weakness, at the expenses of a substantial requirement on the potential regularity. To be specific, we assume that
 \begin{enumerate}[(v1)]
 \setcounter{enumi}{3}
 \setlength{\itemsep}{1.5pt}
\item $v\in C^1(\R)$ with $v'\in L^2(\R)$.  \label{v4}
 \end{enumerate}

\begin{proposition}
\label{prop:fh}
Under assumptions (v\ref{v3}) and (v\ref{v4}), no eigenvalue branch $\epsilon_n(\cdot)$ is constant.
\end{proposition}
\begin{proof}
Denoting by $\tilde\phi_n$ the eigenfunctions of the operator \eqref{eq:fiberTransInvariant2} which are just the function $\phi_n$ of Remark~\ref{rem:efs} with the shifted argument, $\tilde{h_v}(p)\tilde\phi_n = \epsilon_n(p)\tilde\phi_n$, we have
$$ 
    \epsilon_n(p)= \int_\R \tilde\phi_n(x,p) \tilde{h_v}(p)\tilde\phi_n(x,p)\, \D x,
$$ 
and using the fact that $\tilde\phi_n$ is real-valued, we get
$$ 
\epsilon'_n(p) = 2\int_\R \Big(\frac{\partial\tilde\phi_n}{\partial p}\Big)(x,p) \tilde{h_v}(p)\,\tilde\phi_n(x,p)
+ \int_\R \tilde\phi_n(x,p) \frac{\D}{\D p}\tilde{h_v}(p)\,\tilde\phi_n(x,p)\, \D x.
$$ 
The first term on the right-hand side is $\epsilon_n(p)\,\frac{\D}{\D p} (\tilde\phi_n,\tilde\phi_n)$, and since $\tilde\phi_n$ is normalized by assumption, it vanishes. Furthermore, the only $p$-dependent term in \eqref{eq:fiberTransInvariant2} is the last one, thus we obtain
\begin{equation} \label{fh_derivative}
\epsilon'_n(p) = -\frac{1}{B}\, \int_\R v'(x-\textstyle{\frac{p}{B}})\, \tilde\phi_n(x,p)^2\,\D x.
\end{equation}
We can use again the convolution argument for the proof of Proposition~\ref{prop:TraInvSmallPert}. Now we do not know $\tilde\phi_n(\cdot,p)$ explicitly and in order to argue as before we need to ascertain that there is no interval on which the Fourier image of $\tilde\phi_n(\cdot,p)^2$ would vanish. Was it the case, the function must have oscillated asymptotically with a frequency bounded below by the size of the gap \cite{EN04}, however, we know that it is nonnegative and by Sturm-Liouville theory it has a finite number of zeros only. It allows us to conclude that $v'=0$, due to the continuity everywhere, and thus $v$ is a constant function. By assumption (v\ref{v3}) it cannot be nonzero and we arrive at a contradiction. This means that the function $\epsilon'_n(\cdot)$ is nonzero which we set out to prove.
\end{proof}

\begin{remark} \label{rem:fh}
This simple trick is attributed to Hellmann (1937) and Feynman (1939) was noted earlier by G\"uttinger and Pauli. In physics literature, not obsessed with assumptions, it is referred to as Feynman-Hellmann theorem, in mathematics this name is often reserved for its finite-dimensional version. \cite[Thm.~1.4.7]{Si}. The fates of the theorem name bearers reflect cruel turns of the 20th century history; while the former had a long and spectacularly dazzling carrier, the latter died at a young age by hand of Stalin's henchmen \cite{PT18}.
\end{remark}

Trying to reduce the regularity requirement on the potential $v$ one can employ a simple comparison of eigenvalues:
\begin{lemma}
\label{lem:minimax}
Let $v_1$ and $v_2$ be two potentials satisfying assumption (v\ref{v1}). If $v_1(x)\le v_2(x)$ holds a.e. in $\R$, then the corresponding eigenvalues satisfy $\epsilon_n(p,v_1)\le \epsilon_n(p,v_2)$ for all $p\in\R$ and $n\in\N_0$. In particular, if the potential $v$ is sign definite, $v=v_+$ or $v=v_-$, we have $\epsilon_n(p,v)\ge B(2n+1)$ and $\epsilon_n(p,v)\le B(2n+1)$, respectively.
\end{lemma}
\begin{proof}
The potential inequality implies the analogous inequality between the respective quadratic forms \eqref{eq:tildeqform}, hence the claim follows for the minimax principle \cite[Thm~XIII.2]{RS}.
\end{proof}

This allows us, to prove the absolute continuity for sign-definite potentials, however, under an additional assumption:
 \begin{enumerate}[(v1)]
 \setcounter{enumi}{4}
 \setlength{\itemsep}{1.5pt}
\item there is an open interval $J\subset\R$ and a nonzero nonnegative function $w\in C^1$ supported in $\bar J$ satisfying $w(x)\le v_+(x)$ or satisfying $w(x)\le v_-(x)$ a.e. \label{v5}
 \end{enumerate}

\begin{proposition}
\label{prop:approx_onesided}
Suppose that potential $v$ satisfying (v\ref{v3}) is sign-definite and (v\ref{v5}) holds for the respective sign, then no eigenvalue branch $\epsilon_n(\cdot,v)$ is constant.
\end{proposition}
\begin{proof}
For the sake of definiteness suppose that $v=v_+$, the other case is analogous. By Lemma~\ref{lem:minimax} we have $\epsilon_n(p,v)\ge \epsilon_n(p,w) \ge B(2n+1)$ for all $p\in\R$ and $n\in\N_0$, however, $\epsilon_n(\cdot,w)$ is nonconstant by Proposition~\ref{prop:fh} and $\lim_{|p|\to\infty} = \epsilon_n(p,w)$ as shown in Proposition~\ref{prop:TraInvSmallPert}, so the claim follows.
\end{proof}

\begin{remark} \label{rem:fh}
We cannot get rid of assumption (v\ref{v5}) because in general one cannot minorize a nonnegative $L^2$ function by a nontrivial continuous one. As an example, consider function $v$ equal to $\sqrt{2}$ at points of the Smith-Volterra-Cantor set (the "fat Cantor" set) \cite{SVC} and zero elsewhere in $\R$. The measure of the set is $\frac12$ so such a function is a normalized element of $L^2(\R)$, however, there is no interval on which it would be nonzero.
\end{remark}

\section{Properties of the spectrum and numerical examples}
\setcounter{equation}{0}

Let us next list some simple properties of the spectrum assuming that $v$ satisfy one of the above conditions guaranteeing the absolute continuity of the spectrum.
\begin{proposition}
\label{prop:weakpot}
If \mbox{$\|v\|_\infty<B$}, all spectral gaps remain open. If, in addition, the potential is sign-definite, $v=v_+$ or $v=v_-$, the same is true if $\|v\|_\infty<2B$.
\end{proposition}
\begin{proof}
By an elementary perturbation argument, $|\epsilon_n(p)-B(2n+1)|\le \|v\|_\infty$ which yields the claim as the Landau level spacing is $2B$; in case of a sign definite potential we use in addition Lemma~\ref{lem:minimax}.
\end{proof}

By \eqref{first_pert} the band width dependence on the coupling constant is dominantly linear in the weak-coupling situation. Beyond this regime this is no longer true, but in the sign-definite case the dependence remains monotonous.
\begin{proposition}
\label{prop:monoton}
Let $v$ be sign-definite, $v=v_+$ or $v=v_-$, then the  band widths of $H_{\lambda V}$ are strictly monotonous in $\lambda$.
\end{proposition}
\begin{proof}
We use again a Feynman-Hellmann type argument. Similarly as in \eqref{fh_derivative} we get
\begin{equation} \label{fh_derivative}
\frac{\D\epsilon_n}{\D\lambda}(p,\lambda v) = \int_\R v(x-\textstyle{\frac{p}{B}})\, \tilde\phi_n(x,p)^2\,\D x,
\end{equation}
and since $\tilde\phi_n(\cdot,p)^2$ is positive a.e. in $\R$, the claim follows.
\end{proof}

The monotonicity does not exclude the possibility that the gaps remain open. Recall that the results of \cite{Gr21} indicate that this might happen if the potential $v$ is replaced by a $\delta$ interaction, even if the proof is missing.

For a regular $v$, however, the situation may be different; let us look at some examples. To illustrate that, consider a family of potentials approximating the $\delta$-interaction, namely
\[
v(x) = \lambda \cdot \frac{1}{\pi} \cdot \frac{a}{x^2 + a^2}
\]
for which the corresponding family of Schr\"odinger operators converges in the norm-resolvent sense as $a\to 0$ to the one the one with $\delta$-interaction of strength $\lambda$ at $x=0\:$ \cite[Sec.~1.5]{AGHH}. Figures 1 and 2 show the dispersion curves in the repulsive and attractive situation, respectively.
\begin{figure}[h!]
\centering
    \includegraphics[width=0.8\textwidth]{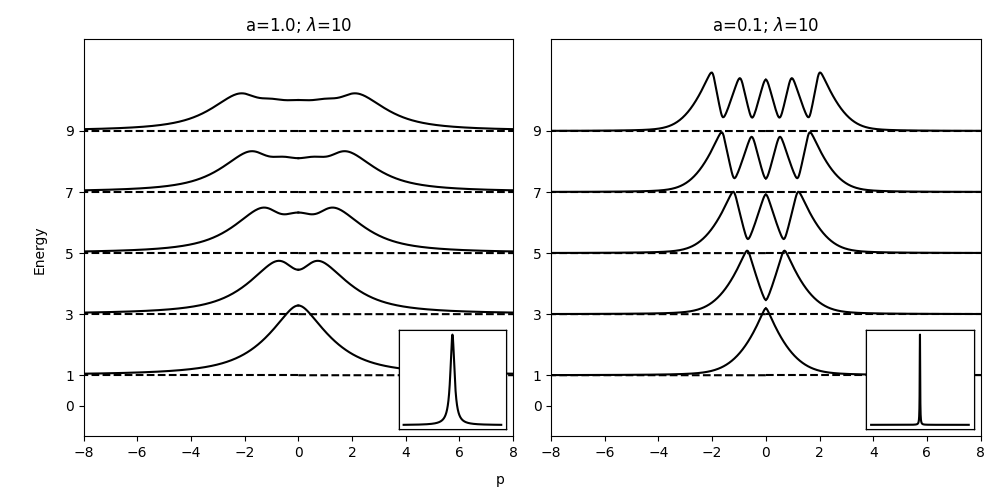}
\caption{Dispersion curves for approximants of a repulsive $\delta$-interaction; the inset shows the potential shape.}
\end{figure}
\begin{figure}[h!]
\centering
    \includegraphics[width=0.8\textwidth]{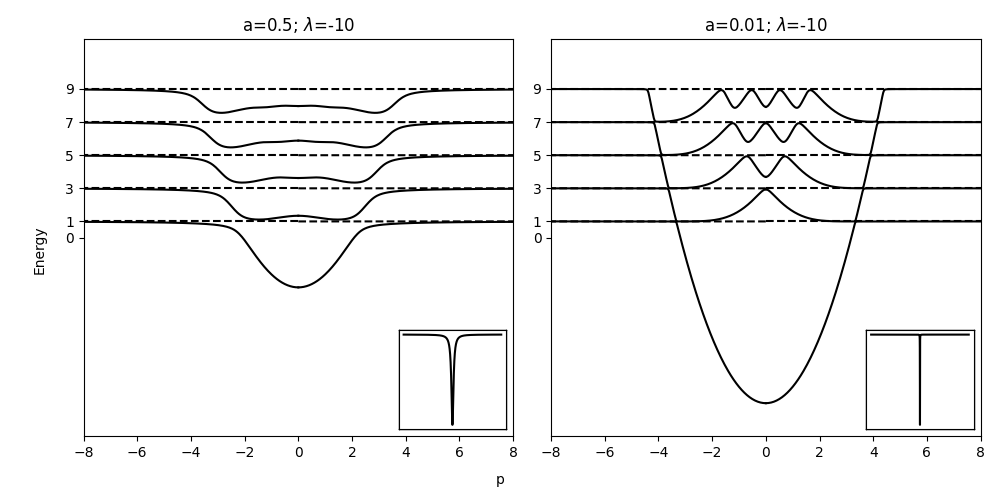}
\caption{Dispersion curves for approximants of an attractive $\delta$-interaction.}
\end{figure}
We see that for small values of $a$ the curves resemble those of the $\delta$-interaction, however, gaps which are open for the latter may close for the approximating potentials.
\begin{figure}[h!]
\centering
    \includegraphics[width=0.8\textwidth]{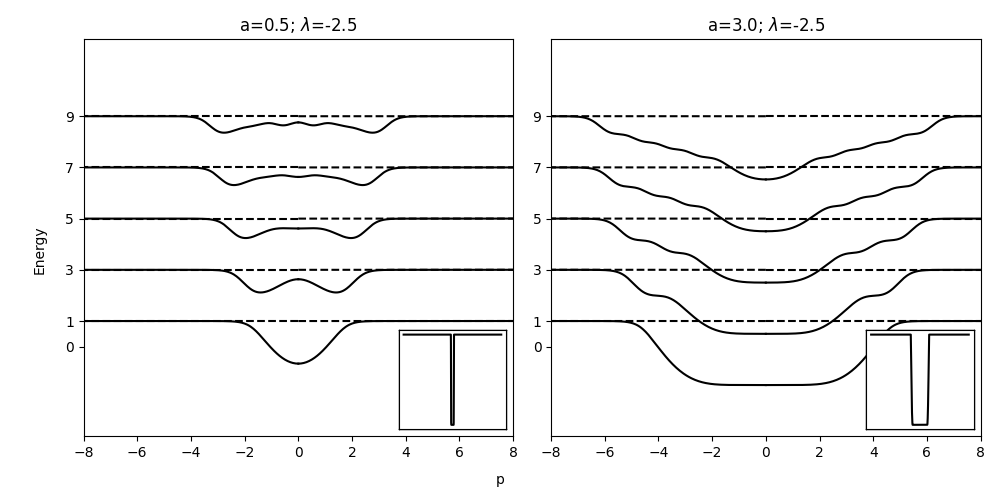}
\caption{Dispersion curves for flat-bottom potential well.}
\end{figure}

Let us next consider a potential well with flat bottom. For numerical reasons, we avoid using a rectangular shape; instead, we smooth the well edges by choosing it as
\[ v(x; a, b, \lambda) = \begin{cases} \lambda & \text{if } |x| < a \\ \lambda \cdot \cos\left(\frac{\pi}{2} \cdot \frac{|x| - a}{(b - a) }\right) & \text{if } a \leq |x| \leq b \\ 0 & \text{otherwise} \end{cases}
\]
The dispersion curves are shown in Figure~3; we see that as the well becomes wide, their central part is close to the appropriately shifted Landau level. Finally, consider a potential without the mirror symmetry, as an example we take
\[
v(x; \lambda, a) = \begin{cases} \lambda \cdot \sin\left(\frac{x}{a}\right) & \text{if } |x| < a\pi \\ 0 & \text{otherwise} \end{cases}
\]
\begin{figure}[h!]
\centering
    \includegraphics[width=0.8\textwidth]{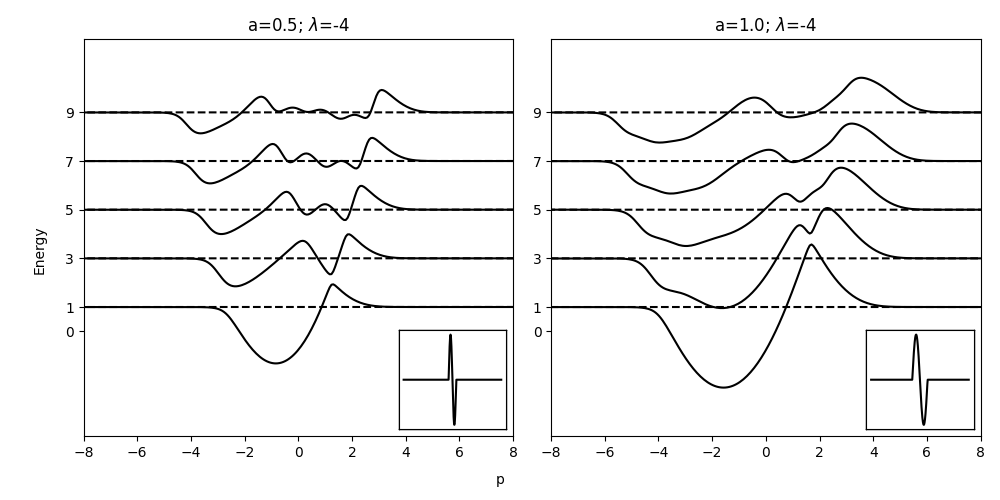}
\caption{Dispersion curves for sign-changing potential.}
\end{figure}
The corresponding dispersion curves are shown in Figure~4, as expected, they are also not mirror symmetric.

\subsection*{Data availability statement}

Data are available in the article.

\subsection*{Conflict of interest}

The authors have no conflict of interest.


\subsection*{Acknowledgments}

The work was supported by the European Union's Horizon 2020 research and innovation programme under the Marie Sk{\l}odowska-Curie grant agreement No 873071.


\end{document}